\documentclass[onecolumn,12pt]{article}
\usepackage{amsfonts}
\usepackage{bbm}
\usepackage{authblk}

\usepackage{geometry}
\usepackage{amsmath}
\usepackage{amssymb}
\usepackage{amsthm}
\usepackage{graphicx}
\usepackage{subfigure}
\usepackage{rotating}
\usepackage{lscape}
\usepackage{paralist}
\usepackage{mathrsfs}
\usepackage{lettrine}
\usepackage[rm,small,compact]{titlesec}
\usepackage{bibentry}
\usepackage[round,authoryear]{natbib}
\usepackage{abstract}
\usepackage{mathptmx}
\usepackage[scaled=0.92]{helvet}
\usepackage{courier}
\usepackage[T1]{fontenc}
\usepackage{enumitem}
\usepackage{titling}
\usepackage{mathtools}
\usepackage{times}
\usepackage{color}
\usepackage{lineno}
\usepackage{longtable}
\usepackage[varg]{txfonts}
\usepackage[titletoc,title]{appendix}
\usepackage{yfonts}
\usepackage{arydshln}
\usepackage{dsfont}
\usepackage{float}

\theoremstyle{plain}
\theoremstyle{definition}

\newtheorem{lemma}{\textsc{Lemma}}[section]
\newtheorem{theorem}{\textsc{Theorem}}[section]
\newtheorem{corollary}{\textsc{Corollary}}[section]

\newtheorem{proposition}{\textsc{Proposition}}[section]

\newtheorem{axiom}{\textsc{Axiom}}[section]
\newtheorem*{theorem*}{\textsc{Theorem}}

\theoremstyle{definition}

\newtheorem*{assertion*}{\textsc{Assertion}}

\newlist{Axiom}{enumerate}{1}
\setlist[Axiom]{label=Axiom A\arabic*.}

\makeatletter
\renewenvironment{proof}[1][\proofname]{\par
\vspace{-10pt}
    \pushQED{\qed}%
    \normalfont \partopsep=\z@skip \topsep=\z@skip
  \trivlist
  \item[\hskip\labelsep
        \itshape
    #1\@addpunct{.}]\ignorespaces
}{%
    \popQED\endtrivlist\@endpefalse
} \makeatother

 \makeatletter
    \renewcommand{\thefigure}{\ifnum \c@section>\z@ \thesection.\fi \@arabic\c@figure}
    \renewcommand{\thetable}{\ifnum \c@section>\z@ \thesection.\fi \@arabic\c@table}
    \makeatother

\renewcommand{\proofname}{\textsc{Proof}.}

\makeatletter 
\@addtoreset{equation}{section}
\makeatother  

\makeatletter
\newcommand{\rmnum}[1]{\romannumeral #1}
\newcommand{\Rmnum}[1]{\expandafter\@slowromancap\romannumeral #1@}
\makeatother

\makeatletter
\newcommand*\bigcdot{\mathpalette\bigcdot@{.5}}
\newcommand*\bigcdot@[2]{\mathbin{\vcenter{\hbox{\scalebox{#2}{$\m@th#1\bullet$}}}}}
\makeatother

\DeclareSymbolFont{euex}{U}{euex}{m}{n}
\DeclareMathSymbol{\varint}{\mathop}{euex}{"52}
\DeclareMathAlphabet{\mathpzc}{OT1}{pzc}{m}{it}

\setenumerate[1]{itemsep=0pt,partopsep=0pt,parsep=\parskip,topsep=1pt}

\begin{document}

\setlength{\abovedisplayshortskip}{5pt}
\setlength{\belowdisplayshortskip}{5pt}
\setlength{\abovedisplayskip}{5pt}
\setlength{\belowdisplayskip}{5pt}

\title{\bf \LARGE{Aggregation of Time Preferences and Comparability of Intergenerational Utilities}}

\author{Wei Ma\thanks{Corresponding author: Center for Economic Research, Shandong University, Jinan, 250100, China. Email: wei.ma@sdu.edu.cn.}}

\date{}
 \maketitle
\begin{onecolabstract}\noindent
Experimental evidence reveals that, in aggregating time preferences, a decision maker exhibits mild aversion to inequality across individuals' evaluations of a stream. Neither the utilitarian nor the maxmin aggregation rule can account for this finding. Building on principles concerning the comparability of intergenerational utilities and intergenerational fairness, this paper axiomatizes a family of aggregation rules that allow the decision maker to display varying degrees of inequality aversion. These rules fit well with experimental data.
\end{onecolabstract}

\section{Introduction}\label{section:Introduction}
In reality, many important economic decisions are made by collectives rather than individuals and have intergenerational ramifications. A good case in point is climate change. Its resolution demands global collaboration. To evaluate an environmental policy, therefore, we need to aggregate time preferences of different individuals, which normally exhibit substantial heterogeneity \citep{Weitzman2001, Falk2018}. On the other hand, experimental evidence indicates that in aggregating time preferences, decision-makers (DM) typically demonstrate a moderate degree of aversion to the welfare inequality across individuals. For instance, in the experiment of \cite{Jackson2014}, almost all the participants display inequality aversion, but for most of them the degree of aversion is not significant.\footnote{It should be stressed that by inequality here we mean inequality across the individuals' evaluations of a stream, not intergenerational inequality. \cite{Jackson2014} fit the model, $a\times \text{mean of payoffs}-(1-a)\times \text{standard deviation of payoffs}$, with each participant's choice data. The value of $a$ providing the best fit is recorded. Depending on whether it is exactly equal to one or less than one, the participant is regarded as  inequality-neutral or  inequality-averse, respectively. \citeauthor{Jackson2014} find that more than 75\% of the participants has a parameter $a$ no less than $0.7$.} What aggregation rules can account for such behavior? What are the rationales for adopting those rules?

The most common rule for aggregating time preferences is the utilitarian rule \citep{Weitzman2001}. It is, however, inequality-neutral. An inequality-averse aggregation rule often discussed in the literature is the maxmin rule, but this rule takes inequality aversion to the extreme. Therefore, neither of these two rules can account for the experimental data of \cite{Jackson2014}. In this paper, we propose several new aggregation rules that allow the DM to exhibit varying degrees of inequality aversion, and provide an axiomatic characterization of these rules.

We follow the social choice theory \citep{Sen1970a, Aspremont1977} by putting forward some axioms capturing intergenerational comparison of utility and intergenerational fairness. In particular, we build on \cite{Chambers2018a} (CE). Assume there are a finite number of individuals, each of them having an exponential discounted utility (EDU). The objects of choice are utility streams over infinitely many generations. Two axioms in CE are concerned with intergenerational comparison of utility. The first specifies that the utilities of different generations are co-cardinal: they have a common origin and a common scale. It is formally expressed as  the DM's preference being invariant under a common positive affine transformation of all generations' utilities. This axiom is not compatible with some well-known models that capture inequality aversion, most notably the mean-variance model. As such, we shall not incorporate this axiom in our paper.

The second axiom, termed IOU, stipulates that the DM's preference is invariant to generation-specific translations of utility origins. This axiom permits intergenerational comparison of utility changes, but prohibits intergenerational comparison of utility levels. Because of the latter, it rules out inequality aversion based on the individuals' utility levels. For this reason, the bulk of the paper is devoted to relaxing IOU in one way or another and examining what aggregation rules arise.

Specifically, we first relax IOU by requiring the DM's preference to be invariant only under a common translation of the utility origins of all generations. For ease of reference, call it COU.  We then introduce another two axioms of CE. One captures intergenerational fairness: the DM prefers to smooth utility streams across generations; the other is a unanimous condition: if every individual prefers one stream to another, so does the DM.  Altogether, these three axioms along with a technical condition imply the DM's preference is a variational preference in the sense of \cite{Maccheroni2006} (Theorem~\ref{Theorem: VP}). This preference includes as special cases a refined version of the mean-variance model \citep{MacCheroni2009} and the multi-utilitarian model \citep{Chambers2018a}.

In Theorem~\ref{Theorem: VP}, each individual's time preference is represented by the average discounted utility function. That is, for an individual with discount factor $\delta$, his  time preference is represented by the function $(1-\delta)\sum_{t=0}^{\infty}\delta^t x_t$ for a stream $(x_0, x_1, \ldots)$. Another widely used representation is by the total discounted utility function $\sum_{t=0}^{\infty}\delta^t x_t$. Although they are equivalent for the individual, the corresponding variational preferences of the DM will be quite different. To see this, note that the variational preference contains the maxmin preference as a special case. Suppose there are two individuals with discount factors $\delta_1$ and $\delta_2$. Then the two preferences, $\min\{(1-\delta_1)\sum_{t=0}^{\infty}\delta_1^t x_t, (1-\delta_2)\sum_{t=0}^{\infty}\delta_2^t x_t\}$ and $\min\{\sum_{t=0}^{\infty}\delta_1^t x_t, \sum_{t=0}^{\infty}\delta_2^t x_t\}$, are far from  identical.

It is then natural to ask under what condition each individual's time preference admits a total discounted utility representation within the DM's variational preference framework. The answer lies in another weakening of IOU: the DM's preference is invariant to translations of the present generation's origin of utility. We refer to this axiom as POU. Replacing COU in Theorem~\ref{Theorem: VP} with POU, we deduce that the DM's preference has a variational preference representation in which each individual's time preference is represented by the total discounted utility function (Theorem~\ref{Theorem: VP total}). POU allows for intergenerational comparison of utility changes and comparison of future generations' utility levels. In the terminology of \cite{Sen1970}, it captures partial comparability of intergenerational utility.

So far, we have been dealing with comparability and fairness of the generations' utilities. However, there is another agent in this context: the individuals, who have differing discount factors. How is this conflict to be resolved? We call an individual unconcerned with respect to two streams if he is indifferent between them. We propose an axiom requiring the DM's ranking of any two streams to be invariant with respect to the evaluations of unconcerned individuals. Augmenting Theorem~\ref{Theorem: VP total} with this axiom yields a weighted utility representation of the DM's preference (Proposition~\ref{Proposition: exponential}). In particular, the DM assigns a weight to each individual, evaluates the individual's total discounted utility via the exponential function $-e^{-t}$, and uses the resulting weighted sum to rank streams.

The function $-e^{-t}$ indicates the DM is inequality-averse. Sometime, it is favourable to have a more flexible function to capture inequality aversion. For this reason, we introduce a third weakening of IOU. That is, if a stream points in a preference-increasing direction when all generations have a common origin of $0$, it continues to do so when the origin of the present generation is shifted while those of all future generations remain fixed at $0$. Substituting this axiom for POU in Proposition~\ref{Proposition: exponential} again delivers a weighted utility representation of the DM's preference, but now each individual's total discounted utility is evaluated via a strictly concave function (Theorem~\ref{Theorem: SOEU total}).

To prove the above results, we adopt a strategy different from that of CE. Specifically, we consider each individual as a possible state of nature. We map each stream to an act defined on the set of states whose consequence in each state is the corresponding individual's evaluation of the stream. We then translate the above axioms into properties about acts. Finally, applying the relevant decision theories under uncertainty will give us the desired results. Our strategy departs from the conventional analogy between intertemporal choice and choice under uncertainty, wherein it is time periods, not individuals, that are treated as states of nature.

To assess how well our aggregation rules can account for the experimental data of \cite{Jackson2014}, we take four models and fit them to the data. The first is the refined mean-variance model mentioned above, the second is a special case of the multi-utilitarian model, the third is the preference in Proposition~\ref{Proposition: exponential}, and the fourth is an instance of the preference in Theorem~\ref{Theorem: SOEU total}. We find their performance is nearly identical, with their scores ranging from $0.8$ to $0.83$. Their fitness is also comparable with that of the mean-standard deviation model adopted in \cite{Jackson2014}, which has a score $0.82$.

\emph{Related Literature}.---This paper contributes to the literature on the aggregation of time preferences. The importance of this aggregation problem stems from the need of a social discount rate to evaluate public policies, such as climate mitigation and natural resource management. There are, however, inherent difficulties in determining the social discount rate \citep{Marglin1963, Feldstein1964, Caplin2004}. What value should be adopted is highly controversial \citep{Stern2007, Nordhaus2007}, and the issue is largely considered ethical \citep{Millner2020}. To address the issue, the social choice approach seeks to derive a social discount rate from individual time preferences under certain normative conditions.

The most widely used condition is the Pareto condition. This condition turns out to be incompatible with the stationarity of the social preference, unless the preference is dictatorial \citep{Zuber2011, Jackson2015}. And it is compatible with time consistency of the social preference only under very special conditions: \cite{Millner2018} show that a weighted utilitarian function can be time consistent only if the weights take a special form \citep{Ma2023a}. 

These negative results lead to two lines of inquiry. One line studies the weakening of the Pareto condition to recover the stationarity of the social preference \citep{Feng2018, Hayashi2021, Billot2025}. The other line retains the condition and examines the consequent aggregation rules \citep{Chambers2018a} and the properties of the corresponding social discount rate \citep{Gollier2005}. Three commonly encountered rules are the multi-utilitarian rule, the maxim rule, and, most notably, the weighted utilitarian rule \citep{Weitzman2001, Freeman2015}. \citet{Chambers2018a} characterize the three rules in which the time preference of each individual is represented by the average discounted utility function. Their work is extended by \cite{Dong-xuan2024a}, who propose, among others, an aggregation method called variational discounting. This method bears a formal resemblance to the one in our Theorem~\ref{Theorem: VP}. But they differ substantially in content (and hence in the axiomatic foundation and proof technique): the cost function in variational discounting is defined on the set of discount factors while the cost function in Theorem~\ref{Theorem: VP} is defined over probability distributions on the set of individuals. Moreover, their motivations are distinct. \cite{Dong-xuan2024a} attempt to find a \lq\lq robust\rq\rq solution to the issue of multiple discount rates. In contrast, the present paper is motivated by the experimental result of  \cite{Jackson2014} and contributes additional aggregation rules that allow a DM to exhibit varying degrees of aversion to welfare inequality across individuals, and in which each individual's time preference can be represented by either the average or total discounted utility function.

The paper is organized as follows. Section~\ref{section:The Setup} describes the setup. Section~\ref{section:Generation} presents the axioms regarding the comparability of intergenerational utilities and intergenerational fairness, and derives the corresponding aggregation rules. Section~\ref{section:Experts} introduces the axioms for resolving individuals' conflicting discount factors and the aggregation rules they imply. Section~\ref{section:Application} applies these aggregation rules to explain the experimental data of \cite{Jackson2014}. Section~\ref{section:Conclusion} concludes.

\section{The Setup}\label{section:The Setup}
Let $\mathbb{T}=\{0,1,2,\ldots \}$ be the set of nonnegative integers, each being interpreted as a generation. The objects of choice are streams, $x=\{x_t\}_{t\in \mathbb{T}}$, where $x_t$ is the utility level for generation $t$. We consider the set  $\ell_{\infty}$ of bounded streams, i.e.
$$\ell_{\infty}=\left\{\{x_t\}_{t\in \mathbb{T}}: \sup_{t\in \mathbb{T}} |x_t|<\infty\right\},$$
which is endowed with the supremum norm $\|x\|_{\infty}=\sup_{t\in \mathbb{T}} |x_t|$. For a scalar $\theta\in \mathbb{R}$, we abuse notation and use it to also denote the constant stream $(\theta, \theta, \ldots)$.

Suppose that there are $N$ individuals in society, each being an exponential discounter. The  total discounted utility function of the individual with discount factor $\delta$ is given by
$$\sum_{t\in \mathbb{T}} \delta^t x_t,$$ 
which, for ease of notation, will be written $\langle \delta, x\rangle$. Let $D=\{\delta_1, \ldots, \delta_N\}\subset (0,1)$ be the set of the individuals' discount factors, where $\delta_i\neq \delta_j$ for $i\neq j$. How should a decision maker (DM) aggregate these discount factors into a single decision criterion? The most common aggregation approach is the utilitarian approach. But the experimental study of \cite{Jackson2014} indicates that besides the aggregate utility, the DM also cares about the dispersion of the individuals' evaluations. Therefore, in this paper, we aim to characterize some aggregation rules that take into account both of them.

\section{Axioms concerning Generations and Aggregation Results}\label{section:Generation}

\subsection{Axioms}\label{section:Axioms}
Suppose that the DM has a preference relation (i.e. a weak order) $\succeq$ on $\ell_{\infty}$. Let $\succ$ and $\sim$ be the asymmetric and symmetric parts of $\succeq$, respectively.  Let $x, y, z$ be any three streams in $\ell_{\infty}$ and $\theta$ a constant stream. In this paper, we consider the following properties on $\succeq$.
\begin{axiom}[D-monotonicity]\label{axiom: D-monotonicity}
(a) If $\langle\delta, x\rangle \geq \langle\delta, y\rangle$ for all $\delta\in D$, then $x\succeq y$; and (b)  if, additionally, $\langle\delta, x\rangle > \langle\delta, y\rangle$ for some $\delta\in D$, then $x\succ y$.
\end{axiom}

\begin{axiom}[Co-cardinality]\label{axiom: Co-cardinality}
For any scalar $a>0$, $x\succeq y$ if and only if (iff) $ax+\theta\succeq ay+\theta$.
\end{axiom}

\begin{axiom}[Invariance with respect to individual origins]\label{axiom: IIOU}
$x\succeq y$ iff $x+z\succeq y+z$.
\end{axiom}

\begin{axiom}[Continuity]\label{axiom: Continuity}
The sets $\{z\in \ell_{\infty}: z\succeq x\}$ and $\{z\in \ell_{\infty}: x\succeq z\}$ are closed.
\end{axiom}

\begin{axiom}[Convexity]\label{axiom: Convexity}
(a) If $x \thicksim y$, then $\alpha x +(1-\alpha)y\succeq x$ for any $\alpha\in [0,1]$; and (b) there exists a pair of streams $x', y'$ satisfying $x' \thicksim y'$ and $\alpha x' +(1-\alpha)y'\succ x'$ for some $\alpha\in (0, 1)$.
\end{axiom}

The first four axioms are drawn from \cite{Chambers2018a} and the last one is a variation of their convexity axiom; we therefore only briefly discuss their interpretation. For a more detailed and critical analysis, see Section~2.4 of \cite{Chambers2018a}. Axiom~\ref{axiom: D-monotonicity} (i.e. parts (a) and (b) together) is usually called the strong Pareto condition, while part (a) alone is called the weak Pareto condition. Axioms~\ref{axiom: Co-cardinality} and \ref{axiom: IIOU} are concerned with intergenerational comparison of utility. The former implies that both absolute utility levels and their differences are intergenerationally comparable, whereas the latter dictates that only differences in utility can be compared. Axioms~\ref{axiom: Co-cardinality} will not be used in this paper. It is presented here only for the sake of comparison with other axioms. Axiom~\ref{axiom: Continuity} is a technical requirement. Finally, Axiom~\ref{axiom: Convexity} is a preference for intergenerational fairness: The DM prefers to smooth utility streams across generations. 

Among the five axioms, Axiom~\ref{axiom: IIOU} is arguably the most controversial. It permits each generation to have its own origin of utility, thereby rendering utility levels incomparable. To relax it, we introduce some weaker axioms. The most obvious one is to assume a common origin of utility across all generations.
\begin{axiom}[Invariance with respect to common origins]\label{axiom: ICOU}
 $x\succeq y$ iff $x+\theta\succeq y+\theta$.
\end{axiom}

This axiom also appears in \cite{Dong-xuan2024a} and \cite{Dong-xuan2025}. As opposed to Axiom~\ref{axiom: IIOU},  it allows for intergenerational comparability of utility levels. To compare these two axioms, consider an example of \cite{Chambers2018a}. Given four utility streams: $x=(10, 8, 0, 0,\ldots)$, $y=(14, 4, 0, 0, \ldots)$, $x'=(10, 1008, 0, 0,\ldots)$, and $y'=(14, 1004, 0, 0, \ldots)$, the latter two streams are obtained from the first two by shifting the origin of generation $1$ by $1000$. In stream $y$, generation $1$ is in a relatively poor state (with a utility of $4$) compared to $y'$. The DM might prefer to transfer wealth from generation $0$ to generation $1$ when $1$ is in a bad state, thus favoring $x$ over $y$. However, she may be unwilling to make the same transfer when generation $1$ is in a good state, thereby preferring $y'$ over $x'$. This natural pair of preferences violates Axiom~\ref{axiom: IIOU}, but is compatible with Axiom~\ref{axiom: ICOU}.

What aggregation rule can arise from the weaker Axiom~\ref{axiom: ICOU}? Let $\Delta(D)$ be the set of  probability measures on $D$. A function  $c: \Delta(D)\rightarrow [0,\infty]$ is said to be nondegenerate if there are at least two distinct elements $\mu\in \Delta(D)$ satisfying $c(\mu)<\infty$, i.e. if the set $\{\mu\in \Delta(D): c(\mu)<\infty\}$ is not a singleton.

\begin{theorem}\label{Theorem: VP}
The following conditions are equivalent:\\
(\rmnum{1}) The relation $\succeq$ satisfies Axioms~\ref{axiom: D-monotonicity}(a), \ref{axiom: Continuity}, \ref{axiom: Convexity}, \ref{axiom: ICOU}.\\
(\rmnum{2}) There exists a grounded, convex, lower semi-continuous, and nondegenerate  function $c: \Delta(D)\rightarrow [0,\infty]$ such that for all $x, y\in \ell_{\infty}$,
\begin{equation}\label{equation: VP}
  x\succeq y \Leftrightarrow \min_{\mu\in \Delta(D)}\left(\sum_{\delta\in D}\Big((1-\delta)\langle \delta, x\rangle\mu(\delta)\Big)+c(\mu)\right) \geq \min_{\mu\in \Delta(D)}\left(\sum_{\delta\in D}\Big((1-\delta)\langle \delta, y\rangle\mu(\delta)\Big)+c(\mu)\right).
\end{equation}
\end{theorem}

Representation~\eqref{equation: VP} can be interpreted as follows. The DM is unsure about the appropriate weights for each individual, and thus entertains all possible weight vectors in $\Delta(D)$. Each such weight vector $\mu$ is associated with a cost $c(\mu)$, with lower costs assigned to those that are more plausible. 

{\bf Remark}. \cite{Dong-xuan2024a} propose the concept of variational discounting, which compares two streams according to the following criterion:
$$x\succeq y \Leftrightarrow \min_{\delta\in [0,1)}\left\{(1-\delta)\langle \delta, x\rangle+c(\delta)\right\} \geq \min_{\delta\in [0,1)}\left\{(1-\delta)\langle \delta, y\rangle+c(\delta)\right\}.$$
The distinction between variational discounting and \eqref{equation: VP} is that the cost function $c$ is defined on $[0,1)$ in the former while it is defined on $\Delta(D)$ in the latter.

Note that in Theorem~\ref{Theorem: VP}, each individual's time preference is represented by the average discounted utility function $(1-\delta)\langle \delta, x\rangle$. An equally popular representation is by the total discounted utility function $\langle \delta, x\rangle$ \citep{Weitzman2001, Gollier2005, Jackson2014}. Then what axioms would imply it? It turns out that it suffices to modify Axiom~\ref{axiom: ICOU} into an invariance axiom with respect to generation $0$'s origin. Specifically, let
$$\ell_{\infty}^0=\{(z_0, z_1,\ldots)\in \ell_{\infty}: z_t=0 \forall t\geq 1\}.$$
\begin{axiom}[Invariance with respect to generation $0$'s origin]\label{axiom: I0OU}
 $x\succeq y$ iff $x+z^0\succeq y+z^0$ for all $z^0\in \ell_{\infty}^0$.
\end{axiom}

Axiom~\ref{axiom: I0OU} treats the present generation (i.e. generation 0) and the future generations (i.e. generations $t$ with $t\geq 1$) asymmetrically. It requires the social preference to be translation invariant with respect to the present generation's origin only. This can be justified from the utility-measurement perspective. For the present generation, we can derive its preference by observing its choices. However, standard microeconomic tools (like von Neumann-Morgenstern utility) generally yield utility measures unique only up to a positive affine transformation. This means the origin of the present generation is an arbitrary choice of our measurement scale, and, as such, a social preference should be invariant to it. However, for future generations, we cannot derive their preferences through observing their choices because they do not yet exist. Instead, we can only evaluate their welfare using some objective and absolute metrics like resource availability, climate stability, and biological survival. These metrics all possess a natural zero.

With Axiom~\ref{axiom: I0OU} taking the place of Axiom~\ref{axiom: ICOU}, we have the following counterpart of Theorem~\ref{Theorem: VP}.
\begin{theorem}\label{Theorem: VP total}
The following conditions are equivalent:\\
(\rmnum{1}) The relation $\succeq$ satisfies Axioms~\ref{axiom: D-monotonicity}(a), \ref{axiom: Continuity}, \ref{axiom: Convexity}, \ref{axiom: I0OU}.\\
(\rmnum{2}) There exists a grounded, convex, lower semi-continuous, and nondegenerate  function $c: \Delta(D)\rightarrow [0,\infty]$ such that for all $x, y\in \ell_{\infty}$,
\begin{equation}\label{equation: VP total}
  x\succeq y \Leftrightarrow \min_{\mu\in \Delta(D)}\left(\sum_{\delta\in D}\langle \delta, x\rangle\mu(\delta)+c(\mu)\right) \geq \min_{\mu\in \Delta(D)}\left(\sum_{\delta\in D}\langle \delta, y\rangle\mu(\delta)+c(\mu)\right).
\end{equation}
\end{theorem}

It is well known that \eqref{equation: VP total} contains as a special case the multi-utilitarian aggregation \citep{Chambers2018a}. Specifically, strengthen Axiom~\ref{axiom: I0OU} to the following one:
\begin{axiom}[Invariance with respect to generation $0$'s origin and scale]\label{axiom: I0OS}
 $x\succeq y$ iff $a x+z^0\succeq a y+z^0$ for all $a>0$  and all $z^0\in \ell_{\infty}^0$.
\end{axiom}

To understand the axiom, note that as argued above, the welfare of the future generations has a natural zero and the DM's preference should not be invariant to its translation. But the unit of future generations' welfare can be changed. For instance, we can use different units of measurement to gauge resource availability. Therefore, the DM's preference should be scale invariant.

In the  terminology of \cite{Sen1970}, Axiom~\ref{axiom: Co-cardinality} corresponds to full comparability of intergenerational utilities and  Axiom~\ref{axiom: IIOU} corresponds to unit comparability. Since Axioms~\ref{axiom: I0OU} and \ref{axiom: I0OS} lie in between those two axioms, they correspond to partial comparability of intergenerational utilities. With Axioms~\ref{axiom: I0OS}, we obtain a characterization of the multi-utilitarian aggregation rule.
\begin{corollary}\label{Corollary: multi-utilitarian}
The following conditions are equivalent:\\
(\rmnum{1}) The relation $\succeq$ satisfies Axioms~\ref{axiom: D-monotonicity}(a), \ref{axiom: Continuity}, \ref{axiom: Convexity}, and \ref{axiom: I0OS}.\\
(\rmnum{2}) There exists a nonempty, convex, closed, and non-singleton subset $\Sigma\subseteq \Delta(D)$ such that for all $x, y\in \ell_{\infty}$,
\begin{equation}\label{equation: multi-utilitarian}
  x\succeq y \Leftrightarrow \min_{\mu\in \Sigma}\left(\sum_{\delta\in D}\langle \delta, x\rangle\mu(\delta)\right) \geq \min_{\mu\in \Sigma}\left(\sum_{\delta\in D}\langle \delta, y\rangle\mu(\delta)\right).
\end{equation}
\end{corollary}

\section{Axioms concerning Individuals and Aggregation Results}\label{section:Experts}

The main axioms in the preceding section are concerned with resolving conflict of different generations' utilities. But there is another agent in the context, viz. the individuals, who have conflicting discount rates. How to resolve this conflict? What aggregation rule can arise accordingly? In this section, we study these questions and assume $N\geq 3$.

We focus on the case where the individuals' time preferences are represented by total discounted utility functions, which will be utilized in the next section. The case where their time preferences are represented by average discounted utility functions can be treated likewise.

Note that Axiom~\ref{axiom: D-monotonicity} is a unanimous condition: when all individuals prefer one stream over another, the DM endorses this preference. However, how should the DM resolve the conflict when the individuals' preferences are not unanimous? The following axiom provides a partial solution. Suppose that there is a set $E^c\subset D$ of individuals who are indifferent between $x$ and $y$. Following \cite{Sen1970a}, we may call them unconcerned individuals with respect to the evaluation of $x$ and $y$. Then it seems natural to require that the welfare levels of the unconcerned individuals have no effect on the ranking of $x$ and $y$. To formalize this, consider two additional streams, $x'$ and $y'$. If every individual in $E=D\backslash E^c$ is indifferent between $x$ and $x'$ and between $y$ and $y'$, while every individual in $E^c$ is indifferent between $x$ and $y$ and between $x'$ and $y'$, it is reasonable to demand  $x\succeq y$ iff $x'\succeq y'$. This is a variant of the separability axiom of \cite{Aspremont1977}.
\begin{axiom}[Invariance with respect to unconcerned individuals]\label{axiom: IUE}
If for any four streams $x, y, x', y'$ and a subset $E\subset D$, $\langle\delta, x\rangle = \langle\delta, x'\rangle$, $\langle\delta, y\rangle = \langle\delta, y'\rangle$ for all $\delta\in E$, and $\langle\delta, x\rangle = \langle\delta, y\rangle$, $\langle\delta, x'\rangle = \langle\delta, y'\rangle$ for all $\delta\in E^c$, then $x\succeq y$ iff $x'\succeq y'$.
\end{axiom}

Armed with this axiom, we get a specific inequality-averse social welfare function. For any $\sigma>0$, let
\begin{equation}\label{equation: exponential}
\xi_{\sigma}(t)=
\begin{cases}
  -\exp(-\frac{t}{\sigma}), & \mbox{if } \sigma<\infty, \\
  t, & \mbox{if }\sigma=\infty.
\end{cases}
\end{equation}
Let $\Delta_{++}(D)=\{\mu\in \Delta(D): \mu\gg0\}$.\footnote{Given $f=(f_1, \ldots, f_N)$ and $g=(g_1, \ldots, g_N)$, $f\geq g$ means $f_i\geq g_i$ for all $i=1, \ldots, N$; $f>g$ means $f\geq g$ but $f\neq g$; and $f\gg g$ if $f_i>g_i$ for all $i=1, \ldots, N$.}  
\begin{proposition}\label{Proposition: exponential}
The following conditions are equivalent:\\
(\rmnum{1}) The relation $\succeq$ satisfies Axioms~\ref{axiom: D-monotonicity}, \ref{axiom: Continuity}, \ref{axiom: Convexity}, \ref{axiom: I0OU}, and \ref{axiom: IUE}.\\
(\rmnum{2}) There exists a probability measure $\nu\in \Delta_{++}(D)$ and $\sigma\in (0, \infty)$ such that for all $x, y\in \ell_{\infty}$, 
\begin{equation}\label{equation: MPtotal}
  x\succeq y \Leftrightarrow \sum_{\delta\in D} \xi_{\sigma}\left(\langle \delta, x\rangle\right)\nu(\delta) \geq \sum_{\delta\in D} \xi_{\sigma}\left(\langle \delta, y\rangle\right)\nu(\delta).
\end{equation}
\end{proposition}

The preference on the right hand side of \eqref{equation: MPtotal} has an equivalent formulation, and is called the multiplier preference \citep{Hansen2001, Strzalecki2011}. Specifically, let $\Delta_{\nu}(D)$ be the set of probability measures in $\Delta(D)$ that are absolutely continuous with respect to $\nu$; define the relative entropy
$$R(\mu||\nu)=
\begin{cases}
\sum_{\delta\in D} \mu(\delta)\ln\frac{\mu(\delta)}{\nu(\delta)}, & \mbox{if } \mu\in  \Delta_{\nu}(D),\\
  \infty, & \mbox{otherwise},
\end{cases}
$$
Then the multiplier preference is a special case of \eqref{equation: VP  total} with $c(\mu)=\sigma R(\mu||\nu)$ for some $\sigma>0$.

Since $\xi_{\sigma}$ is concave, the DM with a multiplier preference is averse to the dispersion of individuals' evaluations of a stream (we shall discuss this point in length in the next section). But there is no reason to insist on a special functional form like \eqref{equation: exponential}: Greater flexibility in functional form is occasionally advantageous. In the following, we characterize a preference with more flexible attitude toward the dispersion of individuals' evaluations.

To this end, we introduce another invariance axiom. To motivate, we notice that Axiom~\ref{axiom: ICOU} has a limitation. Specifically, normalize the subsistence utility level to zero. The DM may prefer $(0,0,0,\ldots)$ to $(-10,10,10,\ldots)$ because generation $0$ falls below subsistence in the latter stream. That is, the DM chooses to sacrifice the future generations' welfare to increase that of generation $0$ when the latter is far below subsistence. However, she may not do so when  generation $0$ is quite rich. For instance,  if the welfare of generation $0$ is increased by $20$, the DM will favor $(10, 10, 10,\ldots)$ over $(20, 0, 0,\ldots)$. To accommodate this pair of preferences, we weaken Axiom~\ref{axiom: I0OU} in the following way.
\begin{axiom}[Invariance of preference-increasing direction with respect to generation 0's origin]\label{axiom: Translation Invariance 0}
If $\alpha x+z^0\succeq z^0$ for some $\alpha>0$ and some $z^0\in \ell_{\infty}^0$,  there exists for any other stream $z^{0'}\in \ell_{\infty}^0$ an $\alpha'>0$ such that $\alpha' x+z^{0'}\succeq z^{0'}$.
\end{axiom}

Recall that Axiom~\ref{axiom: I0OU} demands invariance of the DM's preference with respect to generation 0's origin. Axiom~\ref{axiom: Translation Invariance 0} can be understood as demanding invariance of the DM's preference-change with respect to generation 0's origin.  Formally, call $x$ a preference-increasing direction at $z^0$ if $\alpha x+z^0\succeq z^0$ for some $\alpha>0$. Then the axiom stipulates that if $x$ is a preference-increasing direction at some $z^0\in \ell_{\infty}^0$, it remains so when the origin of generation $0$ is shifted. 

Let us examine how Axiom~\ref{axiom: Translation Invariance 0} can accommodate the previous example. Specifically, let $x=(-10,10,10,\ldots)$, $z^0=(20, 0, 0, \ldots)$, and $\alpha=1$. Then we have $\alpha x+z^0\succeq z^0.$ However, if we shift $z^0$ to $z^{0'}=(0, 0, 0, \ldots)$, Axiom~\ref{axiom: Translation Invariance 0} does not imply $x\succeq 0$; rather, it only  entails $\alpha' x\succeq 0$ for some $\alpha'>0$. To see how this condition may hold, take $\alpha'=0.1$, so that $\alpha' x=(-1,1,1,\ldots)$. The DM prefers $0$ to $x$ because generation $0$ falls substantially below subsistence in $x$ and might not survive. In the scaled stream $\alpha' x$, generation $0$ still falls below subsistence, but to a much lesser extent than in $x$, and therefore can very much likely survive the minor hardship. Given this, the DM may be willing to trade off generation $0$'s minor hardship for a better life of the subsequent generations, and thus favors $\alpha' x$ over $0$.

With Axiom~\ref{axiom: Translation Invariance 0}, we obtain a generalization of Proposition~\ref{Proposition: exponential}.
\begin{theorem}\label{Theorem: SOEU total}
Assume $N\geq 3$. The following conditions are equivalent:\\
(\rmnum{1}) The relation $\succeq$ satisfies Axioms~\ref{axiom: D-monotonicity}, \ref{axiom: Continuity}, \ref{axiom: Convexity},  \ref{axiom: IUE}, and \ref{axiom: Translation Invariance 0}.\\
(\rmnum{2}) There exists a strictly concave, strictly increasing, and differentiable function $u$ and a probability measure $\mu\in \Delta_{++}(D)$ such that for all $x, y\in \ell_{\infty}$,
\begin{equation}\label{equation: SOEU  total}
  x\succeq y \Leftrightarrow \sum_{\delta\in D} u\left(\langle \delta, x\rangle\right)\mu(\delta) \geq \sum_{\delta\in D} u\left(\langle \delta, y\rangle\right)\mu(\delta).
\end{equation}
\end{theorem}

\section{Application}\label{section:Application}
In this section, we apply the aggregation rules in the preceding sections to explain the experimental data of \cite{Jackson2014}. We first briefly describe their experimental setup and then present our result.

In \citeauthor{Jackson2014}'s experiment, subjects make a series of decisions over two consumption streams. They act as social planners whose choices determine the payoff profile for a group of three other individuals. For instance, one of the decisions is to choose between the streams $C=(105, 0, 0)$ and $C'=(0, 160, 0)$. In each group, two members are assigned a discount factor explicitly and are assumed to conform with exponentially discounted utility theory. The third member, however, is not assigned a discount factor, his payoff being given directly and being the same across the two choices. Then if member $1$ is assigned a discount factor of $0.2$, member $2$ a discount factor of $0.9$, and member $3$ the payoff of $80$, the two streams $C$ and $C'$ lead, respectively, to two payoff vectors $(105, 105, 80)$ and $(32, 144, 80)$. The subjects are asked to choose between such payoff vectors.

Totally, $60$ subjects participate the experiment, $36$ of them completing $14$ choices and $24$ completing $38$ choices.  As noted above, in making the choices, the subjects care not only the sum but also the distribution of payoffs. Now let us examine how well our aggregation rules fit their experimental data. Since the  third member in each group is not assigned a discount factor, we cannot apply the rules for aggregating average discounted utilities. To apply the rules for aggregating total discounted utilities, we replace the term $\langle \delta, x\rangle$ in \eqref{equation: VP total} and \eqref{equation: SOEU  total} with the individuals' payoffs. To illustrate, consider \eqref{equation: VP total}. The utility level of the payoff vector $(105, 105, 80)$ is given by
$$\min_{\mu\in \Delta(D)}\left\{(105, 105, 80) \cdot \mu+c(\mu)\right\},$$
where the symbol, $\cdot$, stands for the inner product of two vectors.

We compare four models. The first model is a special case of \eqref{equation: VP total}.  Define
\begin{equation}\label{equation: G}
G(\mu||\nu)=
\begin{cases}
 \frac{1}{2} \sum_{\delta\in D} \left(\frac{\mu(\delta)}{\nu(\delta)}-1\right)^2\nu(\delta), & \mbox{if } \mu\in  \Delta_{\nu}(D),\\
  \infty, & \mbox{otherwise}.
\end{cases}
\end{equation}
By taking $c$ in \eqref{equation: VP total} to be $\sigma G$, we get the monotone mean-variance preference \citep{MacCheroni2009}, written $\succeq^{mv}$. That is, $x \succeq^{mv} y$ iff $U_{mv}(x)\geq U_{mv}(y)$, where
\begin{equation}\label{equation: mean-variance}
U_{mv}=\min_{\mu\in \Delta(D)}\left(\sum_{\delta\in D}\langle \delta, x\rangle\mu(\delta)+\sigma G(\mu||\nu)\right), \sigma\geq 0.
\end{equation}
We take the first model to be this preference $\succeq^{mv}$, and the second one to be the multiplier-preference \eqref{equation: MPtotal}. 

The third model is a special case of \eqref{equation: multi-utilitarian}. In practice, it is not clear how to specify the set $\Sigma$. But, according to \citet[][Theorem~3.7]{Hinojosa2008}, the multi-utilitarian preference includes as a special case the single-parameter generalization of the Gini utility function \citep{Donaldson1980}:
\begin{equation}\label{equation: generalized Gini}
  U_G^\sigma(x)=\frac{\sum_{i=1}^{N} [i^\sigma-(i-1)^\sigma]\tilde{d}_i}{N^\sigma},
\end{equation}
where $\sigma\geq 1$ and $(\tilde{d}_i)_{i=1}^N$ is a permutation of $(\langle \delta_i, x\rangle)_{i=1}^N$ such that $\tilde{d}_1\geq \tilde{d}_2\geq\cdots\geq \tilde{d}_N$. Note that $U_G^1$ is simply the utilitarian rule,  $U_G^2$ corresponds to the Gini index, and $U_G^\infty(x)=\min_i \langle \delta_i, x\rangle$. We take model $3$ to be \eqref{equation: generalized Gini} and refer to it as the Gini-model. The last one is an instance of \eqref{equation: SOEU  total} in which
\begin{equation}\label{equation: CRIA}
u_\sigma(t)=\begin{cases}
     \frac{t^{1-\sigma}}{1-\sigma}, & \mbox{if } \sigma\geq 0, \sigma\neq 1, \\
       \log(t), & \mbox{if } \sigma= 1.
     \end{cases}
\end{equation}
By \cite{Atkinson1970}, this model has a constant relative inequality aversion (CRIA), and so we shall call it CRIA-model. In the multiplier-preference and CRIA models, we take $\nu$ to be the uniform distribution, i.e. $\nu=(1/3, 1/3, 1/3)$. 

To assess the models, we follow the method of \cite{Jackson2014}. To be specific, take the mean-variance model for example. For each subject, the score of the model with a given value of $\sigma$ is the proportion of the decisions that are predicted by the model and coincide with the subject's observed choices. We find the value of $\sigma$ that produces the highest score for each subject, and refer to that score as the individual score of the mean-variance model for the subject. Then the overall score of the model is the average over all individual scores. The overall scores of the four models are presented in Table~\ref{table: score}.
\begin{table}[htbp]
  \centering
   \caption{Overall Scores of the Four Models}
    \setlength{\tabcolsep}{15pt}
 \begin{tabular}{ccccc}
\hline\hline
 Model & Mean-variance & Multiplier-preference& Gini & CRIA\\
\hline
Overall score &0.80  & 0.82 & 0.83&0.81 \\
\hline
\end{tabular}
\label{table: score}
\end{table}

Table~\ref{table: score} indicates that there is little difference in performance among the four models, and their performance is also comparable to the mean-standard deviation model of \cite{Jackson2014} (whose overall score is $0.82$). Given that the Gini model slightly outperforms the other three alternatives, we adopt it as a framework to further examine the data. By  \citet{Donaldson1980}, from the function $U_G^\sigma$ we can define two measures of inequality:
\begin{align*}\label{equation: inequality measure}
  I_R^\sigma(x) & =1-\frac{U_G^\sigma(x)}{\Gamma(x)}, \\
  I_A^\sigma(x) & =\Gamma(x)-U_G^\sigma(x),
\end{align*}
where $\Gamma(x)=\frac{1}{N}\sum_{i=1}^{N}\langle \delta_i, x\rangle$. $I_R^\sigma$ is called a relative measure of inequality and $I_A^\sigma$ an absolute measures of inequality. For each subject $i$, we find the value of $\sigma$, denoted $\sigma_i$, that is associated with his individual score of the Gini model. Using $\sigma_i$, we compute the subject's average aversion to inequality:
$$\bar{I}_R^i=\sum_{x\in \Omega_i} I_R^{\sigma_i}(x), \bar{I}_A^i=\sum_{x\in \Omega_i} I_A^{\sigma_i}(x),$$
where $\Omega_i$ is the set of all payoff vectors involved in subject $i$'s choice problems. The distributions of the subjects' average inequality aversion are displayed in Figure~\ref{fig:distribution of IA}. 
\begin{figure}[H]
  \centering
  \includegraphics[bb=25 3 557 317, scale=0.6]{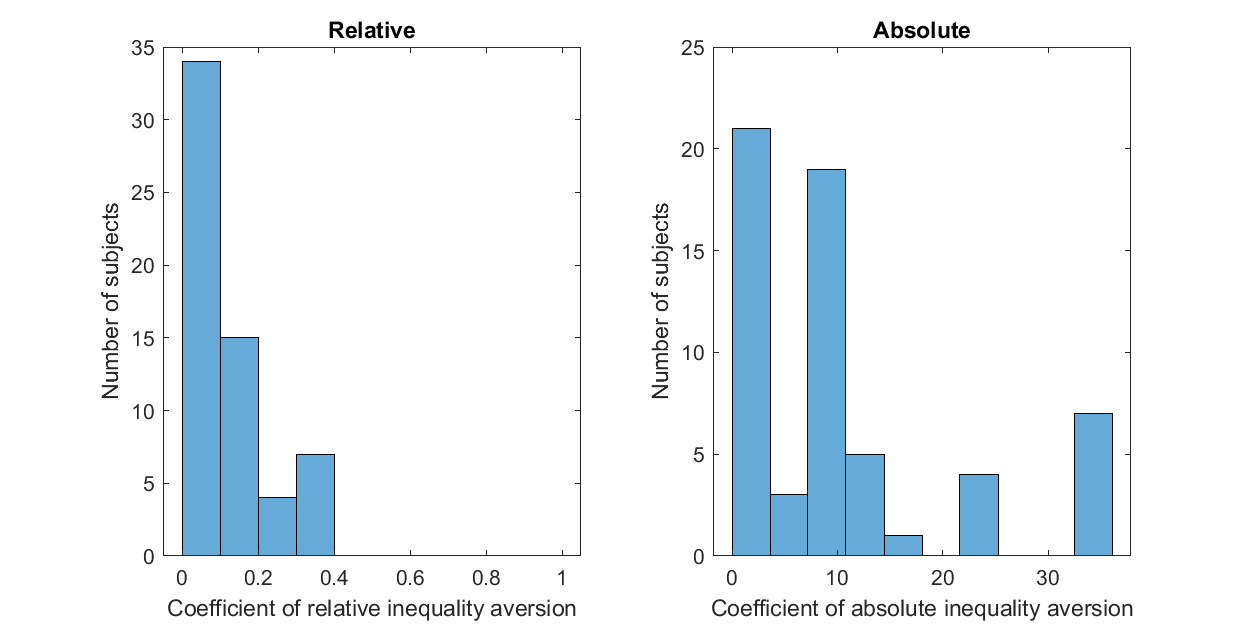}
  \caption{Distributions of the Subjects' Relative and Absolute Inequality Aversion}
  \label{fig:distribution of IA}
\end{figure}
From the left panel we can see that $34$ of the $60$ subjects have a coefficient of relative inequality aversion below $0.1$ and $49$ have a coefficient less than $0.2$. The right panel indicates that $43$ of the $60$ subjects have a coefficient of absolute inequality aversion below $10$. This result corroborates the finding of \cite{Jackson2014} that most subjects in their experiment exhibit inequality aversion but only to a limited degree.

\section{Conclusion}\label{section:Conclusion}
People exhibit substantial heterogeneity in their time preferences, which can generate considerable inequality in their evaluations of the same utility stream. Experimental evidence indicates that a DM is averse to such inequality, but only mildly so. This phenomenon cannot be explained by the two most commonly used aggregation rules: the utilitarian and the maxmin. In this paper, we propose new aggregation rules that allow the DM to display varying degrees of inequality aversion. The new axiom underlying these rules is the invariance of the DM's preference with respect to the origin and scale of the present generation's utility. The predictions of these rules are in good agreement with experimental data.

For future research, note that we have assumed throughout this paper that all individuals are exponential discounters. This assumption is incompatible with many experimental and empirical findings \citep{Frederick2002}. To what extent can it be relaxed? Along this line of inquiry, \cite{Dong-xuan2025} study the aggregation of quasi-hyperbolic individuals. Although quasi-hyperbolic discounting has been quite successful in explaining many important economic phenomena \citep{Laibson}, it falls short of accounting for many other findings---for instance, that present bias becomes much weaker when uncertainty is introduced into the immediate reward \citep{Halevy2008}. This suggests the need to study the aggregation of more general time preferences, such as time-separable preferences.

\appendix
\begin{appendices}
\numberwithin{equation}{section}
\section{Proofs}
The proofs of Theorems~\ref{Theorem: VP} and \ref{Theorem: VP total} are quite similar. Since our main focus in this paper is on the case where the individuals' time preferences are represented by total discounted utility functions, we present the proof of Theorem~\ref{Theorem: VP total} in detail and indicate only how to adapt this argument to arrive at a proof of Theorems~\ref{Theorem: VP}.

\subsection{Proof of Theorem~\ref{Theorem: VP total}}\label{section: Proof of Theorem VP total}

The main idea in the proof of the theorem is to apply the relevant representation result in decision making under uncertainty. To this end, we transform each stream $x\in \ell_{\infty}$ into an act defined on $D$.  Let $\mathfrak{F}$ be the set of all functions from $D$ to $\mathbb{R}$. We refer to the functions in $\mathfrak{F}$ as acts, and identify $a\in \mathbb{R}$ with a constant act $f(\delta)=a$ for all $\delta\in D$. Endow $\mathfrak{F}$ with the Euclidean norm. 

Define a mapping $\phi: \ell_{\infty} \rightarrow \mathbb{R}^N$ such that $\phi(x)=(\langle \delta_i, x\rangle)_{i=1}^N$.
\begin{lemma}\label{lemma: image}
(\rmnum{1}): $\phi$ is continuous and linear.\\
(\rmnum{2}):  $\phi(x_0, 0, 0, \ldots)=x_0$ for any $x_0\in \mathbb{R}$.\\
(\rmnum{3}): $\mathfrak{F}=\phi(\ell_{\infty})$.
\end{lemma}
\begin{proof}
Parts (\rmnum{1}) and (\rmnum{2}) can be verified directly from the definition of $\phi$. For part (\rmnum{3}), because $D\subset (0,1)$, we have $\phi(\ell_{\infty}) \subseteq \mathfrak{F}$. To show the reverse inclusion, take any $f=(f_1, \ldots, f_N)\in \mathfrak{F}$. It suffices to demonstrate $f=\phi(x)$ for some $x\in \ell_{\infty}$. To this end, consider the linear system
\begin{equation}\label{eq: determine x from f}
\begin{bmatrix}
  1 & \delta_1 & \cdots &  \delta_1^{N-1} \\
   1 & \delta_2 & \cdots &  \delta_2^{N-1} \\
  \vdots &   \vdots & \ddots &   \vdots \\
  1 & \delta_N & \cdots &  \delta_N^{N-1} 
\end{bmatrix}
\begin{bmatrix}
  x_1 \\
  x_2 \\
  \vdots \\
  x_N 
\end{bmatrix}=
\begin{bmatrix}
f_1 \\
f_2\\
    \vdots \\
f_N
\end{bmatrix}.
\end{equation}
Since the coefficient matrix of the system is a Vandermonde matrix, which, as $\delta_i\neq \delta_j$ for $i\neq j$, is nonsingular, the system has a unique solution, denoted $(x_1^*, \ldots, x_N^*)$. Then we have $f=\phi(x)$ by letting $x=(x_1^*, \ldots, x_N^*, 0, 0, \ldots)$.
\end{proof}

Let $\succeq$ be a preference relation on $\ell_{\infty}$ which satisfies Axiom~\ref{axiom: D-monotonicity}(a). With this axiom, $\phi(x)=\phi(y)$ implies $x\sim y$. Therefore, the binary relation $\succeq^*$ on $\mathfrak{F}$ defined by
$$f \succeq^* g\text{ if }x\succeq y, \text{ where }x\in \phi^{-1}(f) \text{ and }y\in \phi^{-1}(g), $$ 
is a well-defined weak order. In what follows, we translate properties of $\succeq$ into corresponding properties of $\succeq^*$, and then apply the relevant results from decision theory under uncertainty to obtain a representation of $\succeq^*$, which in turn delivers the desired representation of $\succeq$.

Now let us prove the theorem. For the implication (\rmnum{2})$\Rightarrow$ (\rmnum{1}), the verification of Axioms~\ref{axiom: D-monotonicity}(a), \ref{axiom: Continuity}, \ref{axiom: Convexity}(a), and \ref{axiom: I0OU} is straightforward. To verify Axiom~\ref{axiom: Convexity}(b), let $\text{dom}(c)=\{\mu\in \Delta(D): c(\mu)<\infty\}$ and
$$U(f)=\min_{\mu\in \Delta(D)}\left(\sum_{\delta\in D} f(\delta)\mu(\delta)+c(\mu)\right).$$
\begin{lemma}\label{lemma: affine and nondegenerate}
$U$ is affine iff $\text{dom}(c)$ is a singleton.
\end{lemma}
\begin{proof}
If $\text{dom}(c)$ is a singleton, let $\text{dom}(c)=\{\mu\}$. Then $U(f)=\sum_{\delta\in D} f(\delta)\mu(\delta)+c(\mu)$, so that $U$ is affine.

Conversely, suppose that $U$ is affine. Assume by way of contradiction that  $\text{dom}(c)$ contains two distinct elements $\mu_1$ and $\mu_2$. Then there exists an $f\in \mathfrak{F}$ such that $f\cdot \mu_1\neq f\cdot \mu_2$. For any $t\in \mathbb{R}$, we have
\begin{align*}
 \psi(t):=U(tf)&=\min_{\mu\in \Delta(D)}\left(t\sum_{\delta\in D} f(\delta)\mu(\delta)+c(\mu)\right) \\
   & \leq  (f\cdot \mu_i)t+c(\mu_i), i=1, 2.
\end{align*}
Since $f\cdot \mu_1\neq f\cdot \mu_2$, it follows that $ \psi$ and hence $U$ cannot be affine.
\end{proof}

Now as $c$ is nondegenerate, $\text{dom}(c)$ is not a singleton and, by Lemma~\ref{lemma: affine and nondegenerate}, $U$ is not affine. Since $U$ is concave, this implies there exist two acts $f$ and $g$ such that $U(h)>\alpha U(f)+(1-\alpha) U(g)$, where $h=\alpha f+(1-\alpha)g$. Let $\hat{f}=f-U(f)\mathbf{1}$, $\hat{g}=g-U(g)\mathbf{1}$, and $\hat{h}=\alpha\hat{f}+(1-\alpha) \hat{g}$, where $\mathbf{1}=(1, 1, \ldots, 1)\in \mathbb{R}^N$, so that $U(\hat{f})=U(\hat{g})=0$ and $U(\hat{h})>0$. This establishes Axiom~\ref{axiom: Convexity}(b).

We turn to the proof of the reverse implication (\rmnum{1})$\Rightarrow$ (\rmnum{2}).
\begin{lemma}\label{lemma: VP*}
 There exists a grounded, convex, and lower semi-continuous function $c: \Delta(D)\rightarrow [0,\infty]$ such that for all $f, g\in \mathfrak{F}$,
\begin{equation}\label{equation: VP*}
  f\succeq^*g \Leftrightarrow \min_{\mu\in \Delta(D)}\left(\sum_{\delta\in D} f(\delta)d\mu(\delta)+c(\mu)\right) \geq \min_{\mu\in \Delta(D)}\left(\sum_{\delta\in D} g(\delta)d\mu(\delta)+c(\mu)\right).
\end{equation}
\end{lemma}
\begin{proof}
The proof is an adaptation of the proof of Theorem~3 of \cite{Maccheroni2006}. For the sake of completeness we present a sketch. The preference $\succeq^*$ is continuous in the sense that the sets $\{\alpha\in [0,1]: \alpha f+(1-\alpha)g\succeq^* h\}$ and $\{\alpha\in [0,1]: h\succeq^*\alpha f+(1-\alpha)g\}$ are closed for all $f, g, h\in \mathfrak{F}$. We prove the former case only; the latter case can be proved likewise. Suppose $\bar{f}^k:=\alpha_k f+(1-\alpha_k)g\succeq^* h$, $k=1, 2, \ldots$, and $\alpha_k\rightarrow \alpha_0$; we show $\bar{f}^0:=\alpha_0 f+(1-\alpha_0)g\succeq^* h$. Let $\bar{x}^k=(\bar{x}^k_1, \ldots, \bar{x}^k_N, 0, 0, \ldots)$ be a stream in $\ell_{\infty}$ such that $(\bar{x}^k_1, \ldots, \bar{x}^k_N)$ solves \eqref{eq: determine x from f} with $f$ on the right hand side replaced by $\bar{f}^k$, $k=0, 1, 2, \ldots$. Then we have $\bar{f}^k=\phi(\bar{x}^k)$, $k=0, 1, 2, \ldots$, and as the Vandermonde matrix in \eqref{eq: determine x from f} is nonsingular, $\bar{f}^k\rightarrow \bar{f}^0$ implies $\bar{x}^k\rightarrow \bar{x}^0$. Let $z\in \phi^{-1}(h)$; then  $\bar{f}^k\succeq^* h$ implies $\bar{x}^k\succeq z$ for $k=1, 2, \ldots$. By Axiom~\ref{axiom: Continuity}, $\bar{x}^0\succeq z$ and hence $\bar{f}^0\succeq^* h$.

By Axiom~\ref{axiom: D-monotonicity}(a), there exist, for each $f\in \mathfrak{F}$, real numbers $\bar{\theta}$ and $\underline{\theta}$ such that $\bar{\theta}\mathbf{1} \succeq^* f \succeq^* \underline{\theta}\mathbf{1}$. By the continuity of $ \succeq^*$, there exists an $\alpha\in [0,1]$ such that $f \thicksim^* \alpha \bar{\theta}\mathbf{1} +(1-\alpha)\underline{\theta}\mathbf{1}:=\theta_f\mathbf{1}$. We claim that $\theta_f$ is unique. To see this, suppose by way of contradiction that there exists a $\theta'\neq \theta_f$ with $f \thicksim^* \theta'\mathbf{1}$. Therefore, $\theta_f\mathbf{1}\thicksim^* \theta'\mathbf{1}$. Then Axiom~\ref{axiom: I0OU} implies that $(\theta_f-\theta')\mathbf{1}\thicksim^* 0$ and $(\theta'-\theta_f)\mathbf{1}\thicksim^* 0$, hence, by induction, that $a(\theta_f-\theta')\mathbf{1}\thicksim^* 0$ and $a(\theta'-\theta_f)\mathbf{1}\thicksim^* 0$ for any natural number $a$. Assume without loss of generality that $\theta_f>\theta'$. Then for any $f\in \mathfrak{F}$, there exists a sufficiently large natural number $a$ such that $a(\theta_f-\theta')>f(\delta)>a(\theta'-\theta_f)$ for all $\delta\in D$. As a consequence of Axiom~\ref{axiom: D-monotonicity}(a), $f\thicksim^* 0$. But this contradicts \ref{axiom: Convexity}(b). Hence, $\theta_f$ must be unique.

Define the function $I: \mathfrak{F}\rightarrow \mathbb{R}$ by $I(f)=\theta_f$. By Axiom~\ref{axiom: D-monotonicity}(a), $f\geq g$ implies $I(f)\geq I(g)$. For any $f\in \mathfrak{F}$ and $\alpha\in (0,1)$, let $\theta_{\alpha f}=I(\alpha f)$. Then by Axiom~\ref{axiom: I0OU}, $\alpha f +(1-\alpha)\theta\mathbf{1} \thicksim^* (\theta_{\alpha f}+(1-\alpha)\theta)\mathbf{1}$ for any $\theta\in \mathbb{R}$. This means $I(\alpha f +(1-\alpha)\theta\mathbf{1})=\theta_{\alpha f}+(1-\alpha)\theta=I(\alpha f)+(1-\alpha)\theta$. Moreover, by Axiom~\ref{axiom: Convexity}, $I(f)=I(g)$ implies $I(\alpha f+(1-\alpha)g)\geq I(f)$. It follows from Lemma~25 of \cite{Maccheroni2006} that $I$ is a concave normalized niveloid.\footnote{For the definitions of a niveloid and a niveloid being  normalized, see Appendix~A of \cite{Maccheroni2006}.} This together with their Lemma~26 implies that  there exists a grounded, convex, and lower semi-continuous function $c: \Delta(D)\rightarrow [0,\infty]$ such that 
 $$I(f)=\min_{\mu\in \Delta(D)}\left(\sum_{\delta\in D}\Big(f(\delta)\mu(\delta)\Big)+c(\mu)\right).$$
This proves the lemma.
\end{proof}

It remains to show $c$ is nondegenerate. By  Axiom~\ref{axiom: Convexity}(b), $I$ is not affine and hence, by Lemma~\ref{lemma: affine and nondegenerate}, $\text{dom}(c)$ is not a singleton or, equivalently, $c$ is nondegenerate.

\subsection{Proof of Theorem~\ref{Theorem: VP}}\label{section: Proof of Theorem VP}
The only difference between Axioms~\ref{axiom: ICOU} and \ref{axiom: I0OU} is that constant streams in the former are replaced by streams in $\ell_{\infty}^0$ in the latter. A key property of $\phi$ used in the above proof is $\phi(z^0)=x_0$ for any $z^0=(x_0, 0, 0,\ldots)\in \ell_{\infty}^0$. In some sense, we can regard $\phi$ as an identity function on $\ell_{\infty}^0$. Define $\phi_a: \ell_{\infty}\rightarrow \mathbb{R}^N$ by $\phi_a(x)=((1-\delta_i)\langle \delta_i, x\rangle)_{i=1}^N$. Then the same property holds also for $\phi_a$ on the set of constant streams; that is, $\phi_a(\theta)=\theta$ for any constant stream $\theta\in  \ell_{\infty}$.  With this property, the proof of Theorem~\ref{Theorem: VP} is the same as that of Theorem~\ref{Theorem: VP  total} by replacing $\phi$ there with $\phi_a$.

\subsection{Proof of  Corollary~\ref{Corollary: multi-utilitarian}}\label{section: Proof of Corollary multi-utilitarian}
\begin{lemma}\label{lemma: C-independence}
Axiom~\ref{axiom: I0OS} is equivalent to the condition that for all $f, g\in \mathfrak{F}$, $\alpha\in [0,1]$, and $c\in \mathbb{R}$,  $f\succcurlyeq^* g$ iff $\alpha f+(1-\alpha)c\mathbf{1}\succcurlyeq^* \alpha g+(1-\alpha)c\mathbf{1}$.
\end{lemma}
\begin{proof}
Assume Axiom~\ref{axiom: I0OS} and take any $f, g\in \mathfrak{F}$ with $f\succcurlyeq^* g$. Then we have $x\succcurlyeq y$, where $x\in \phi^{-1}(f)$ and $y\in \phi^{-1}(g)$. For $\alpha\in [0,1]$ and $c\in \mathbb{R}$, we have, under Axiom~\ref{axiom: I0OS}, $\alpha x+z^0 \succcurlyeq \alpha y+z^0$, where $z^0=((1-\alpha)c, 0, 0, \ldots)$. Since $\phi(\alpha x+z^0)=\alpha f+(1-\alpha)c\mathbf{1}$ and $\phi(\alpha y+z^0)=\alpha g+(1-\alpha)c\mathbf{1}$, we get $\alpha f+(1-\alpha)c\mathbf{1}\succcurlyeq^* \alpha g+(1-\alpha)c\mathbf{1}$. The reverse direction can be proved likewise.
\end{proof}

In other words, Axiom~\ref{axiom: I0OS} implies $\succcurlyeq^*$ satisfies the C-independence axiom of \cite{Gilbola1989}. Then Corollary~\ref{Corollary: multi-utilitarian} is a joint consequence of their Theorem~1 and Theorem~\ref{Theorem: VP total} of the present paper \citep{Maccheroni2006}.

\subsection{Proof of Proposition~\ref{Proposition: exponential}}\label{section: proof of Proposition exponential}
 For any two acts $f, g\in \mathfrak{F}$ and a subset $E\subseteq D$, $f_Eg$ denotes the act which equals $f(\delta)$ for $\delta\in E$ and equals $g(\delta)$ for $\delta\in E^c$. We say $\succeq^*$ satisfies the Sure-thing principle if for any acts $f, g, h, h'\in \mathfrak{F}$ and a subset $E\subset D$, $f_E h \succeq^* g_Eh$ implies $f_E h' \succeq^* g_Eh'$.

\begin{lemma}\label{lemma: sure-thing}
Axiom~\ref{axiom: IUE} is equivalent to $\succeq^*$ satisfying the Sure-thing principle.
\end{lemma}
\begin{proof}
Assume Axiom~\ref{axiom: IUE} and take any $E=\{\delta_1, \ldots, \delta_M\}\subseteq D$. Let $\mathbf{M}(E)$ be the Vandermonde matrix:
$$\mathbf{M}(E)=
\begin{bmatrix}
  1 & \delta_1 & \cdots & \delta_1^{N-1} \\
  1 & \delta_2 & \cdots & \delta_2^{N-1} \\
  \vdots &  \vdots &  \ddots&  \vdots \\
1 & \delta_M & \cdots & \delta_M^{N-1}
\end{bmatrix},
$$
and similarly for $\mathbf{M}(E^c)$.  For any $f, g, h\in \mathfrak{F}$, consider the following linear system of equations:
$$
\begin{bmatrix}
  \mathbf{M}(E) & 0 \\
  0& \mathbf{M}(E) \\
  \mathbf{M}(E^c) & 0 \\
  0& \mathbf{M}(E^c) \\
\end{bmatrix}
\begin{bmatrix}
\bar{x}^\top\\
\bar{y}^\top
 \end{bmatrix}=
 \begin{bmatrix}
   f|_E \\
   g|_E \\
   h|_{E^c} \\
   h|_{E^c}
 \end{bmatrix}
$$
where $\bar{x}=(x_0, \ldots, x_{N-1})$ and $\bar{y}=(y_0, \ldots, y_{N-1})$ and the symbol, $\top$, stands for the transpose of a matrix. Since the coefficient matrix of the system is nonsingular,  it has a unique solution, denoted $(\bar{x}^*, \bar{y}^*)$. Let $x=(\bar{x}^*, 0, 0, \ldots)$ and $y=(\bar{y}^*, 0, 0, \ldots)$. Similarly, for any $h'\in \mathfrak{F}$, let $(\bar{x}', \bar{y}')$ be the solution to the above linear system with $h$ on the right hand side replaced by $h'$, and $x'=(\bar{x}', 0, 0, \ldots)$ and $y'=(\bar{y}', 0, 0, \ldots)$. Then we have $\phi(x)=f_Eh$, $\phi(y)=g_Eh$, $\phi(x')=f_Eh'$, and $\phi(y')=g_Eh'$. By  Axiom~\ref{axiom: IUE} and the definition of $\succeq^*$, it follows that $\succeq^*$ satisfies the Sure-thing principle.

To prove the reverse implication, suppose that $\succeq^*$ satisfies the Sure-thing principle. Take any four streams $x, y, x', y'$ and a subset $E\subset D$ such that $\langle\delta, x\rangle = \langle\delta, x'\rangle$, $\langle\delta, y\rangle = \langle\delta, y'\rangle$ for all $\delta\in E$ and $\langle\delta, x\rangle = \langle\delta, y\rangle$, $\langle\delta, x'\rangle = \langle\delta, y'\rangle$ for all $\delta\in E^c$. Let $f=\phi(x)$, $g=\phi(y)$, $f'=\phi(x')$, and $g'=\phi(y')$, so that $f=f'$, $g=g'$ on $E$, and $f=g$, $f'=g'$ on $E^c$. It follows from the Sure-thing principle that $f\succeq^* g$ iff $f'\succeq^* g'$, hence that $x\succcurlyeq y$ iff $x'\succcurlyeq y'$. 
\end{proof}

With Lemma~\ref{lemma: sure-thing}, Proposition~\ref{Proposition: exponential} is a consequence of the result of \citet{Strzalecki2011}. Specifically, for the implication (\rmnum{1})$\Rightarrow$ (\rmnum{2}), by Theorem~\ref{Theorem: VP total},  $\succcurlyeq^*$ can be represented by \eqref{equation: VP*}. Axioms~\ref{axiom: D-monotonicity}(b) implies every state $\delta\in D$ is non-null. By Lemma~\ref{lemma: sure-thing} and Theorem~1 of \citet{Strzalecki2011}, $\succcurlyeq^*$ can be represented by
\begin{equation*}
  f\succeq^*g \Leftrightarrow \min_{\mu\in \Delta(D)}\left(\sum_{\delta\in D} f(\delta)\mu(\delta)+\sigma R(\mu||\nu)\right) \geq \min_{\mu\in \Delta(D)}\left(\sum_{\delta\in D} g(\delta)\mu(\delta)+\sigma R(\mu||\nu)\right).
\end{equation*}
By the variational formula \citep[][Proposition 1.4.2, pp. 33--34]{Dupuis1997}, we have
$$\min_{\mu\in \Delta(D)}\left(\sum_{\delta\in D} f(\delta)d\mu(\delta)+\sigma R(\mu||\nu)\right) =\xi_\sigma^{-1}\left(  \sum_{\delta\in D}\xi_{\sigma}\left(\langle \delta, x\rangle\right)\nu(\delta)\right).$$
From this equation and Axioms~\ref{axiom: D-monotonicity}(b) and \ref{axiom: Convexity}(b), it follows that $\nu \in \Delta_{++}(D)$ and $\sigma<\infty$. This completes the proof of (\rmnum{1})$\Rightarrow$ (\rmnum{2}). The reverse implication can be proved likewise and is omitted here.

\subsection{Proof of Theorem~\ref{Theorem: SOEU total}}\label{section: Proof of Theorem SOEU total}

Let us start with the implication (\rmnum{1}) $\Rightarrow$ (\rmnum{2}). First, by Axioms~\ref{axiom: D-monotonicity} and \ref{axiom: Continuity}, it follows from the proof of Lemma~\ref{lemma: VP*} that  there exist for each $f\in \mathfrak{F}$ a real number $\theta_f$ such that $f\thicksim^* \theta_f \mathbf{1}$ and that the function $I: \mathfrak{F}\rightarrow \mathbb{R}$ defined by $I(f)=\theta_f$ is normalized, continuous, and monotone.

Third, we demonstrate the following lemma. 
\begin{lemma}\label{lemma: ambiguity aversion}
There exists $\mu\in \Delta_{++}(D)$ such that for all $f\in \mathfrak{F}$ and $\theta\in \mathbb{R}$, $f\succeq^*  \theta\mathbf{1}$ implies $\mu\cdot f\geq \theta$.
\end{lemma}
\begin{proof}
Fix $\theta\in \mathbb{R}$ and consider the set
$$B(\theta)=\{f\in \mathfrak{F}: I(f)\geq I(\theta\mathbf{1})\}.$$
By Axiom~\ref{axiom: D-monotonicity}, $B(\theta)$ has a nonempty interior and $\theta\mathbf{1}$ is its boundary point. By Axiom~\ref{axiom: Convexity}, $B(\theta)$ is convex. Then it follows from the separating hyperplane theorem \citep[see, e,g., ][Lemma~7.7, p. 259]{Aliprantis2006a} that there exists $\mu\in \mathbb{R}^N$ such that $\mu\cdot (f-\theta\mathbf{1})\geq 0$ for all $f\in B(\theta)$. Again by Axiom~\ref{axiom: D-monotonicity}, we have $\mu\gg 0$. Normalize $\mu$ such that $\mu\in \Delta(D)$.
For any $\theta'\in \mathbb{R}$, let
$$\Pi(\theta')=\{\nu\in \Delta(D): \nu\cdot (f-\theta'\mathbf{1})\geq 0 \text{ for all }f\succeq \theta'\mathbf{1}\},$$
so that $\mu\in \Pi(\theta)$. It suffices to show  $\mu\in \Pi(\theta')$  for all $\theta'\in \mathbb{R}$. As a consequence of Axiom~\ref{axiom: Translation Invariance 0}, we get the axiom of translation invariance at certainty of \cite{Rigotti2008}: For all $f\in \mathfrak{F}$ and $c\in \mathbb{R}$, if $\alpha f+c\mathbf{1}\succcurlyeq^* c\mathbf{1}$ for some $\alpha>0$, then there exists for any other $c'\in \mathbb{R}$ an $\alpha'>0$ such that $\alpha' f+c'\mathbf{1}\succcurlyeq^* c'\mathbf{1}$. It then follows from their Proposition~8 that 
\begin{equation}\label{equation: Pi theta constant}
\Pi(\theta')=\Pi(\theta) \text{ for all }\theta'\in \mathbb{R},
\end{equation}
hence that $\mu\in \Pi(\theta')$.
\end{proof}

Combining the Sure-thing principle and Lemma~\ref{lemma: ambiguity aversion}, it follows from the proof of Theorem~2 of  \cite{Grant2009} that there exists a continuous, concave, and strictly increasing function $u$ on $\mathbb{R}$ such that $\succcurlyeq^*$ is represented by
$$J(f)= \sum_{i=1}^{N}\mu_i u(f_i).$$
It remains to show $u$ is differentiable and strictly concave. To show the differentiability of $u$, let $\partial u(\theta)$ denote the superdifferential of $u$ at $\theta$. Since $u$ is concave, there exists $\theta_0\in \mathbb{R}$ such that $\partial u(\theta_0)$ is a singleton. By Lemma~1 of  \cite{Rigotti2008}, 
$$\Pi(\theta)=\left\{\frac{(\mu_1 q_1, \ldots, \mu_N q_N)}{\mu_1 q_1+\cdots+ \mu_N q_N}: q_i\in \partial u(\theta) \text{ for all }i=1, \ldots, N\right\},$$
so that $\Pi(\theta_0)$ is a singleton. This along with \eqref{equation: Pi theta constant} implies $\Pi(\theta)$ and hence $\partial u(\theta)$ must be a singleton for all $\theta\in \mathbb{R}$. As a consequence, $u$ is differentiable everywhere on $\mathbb{R}$.

For the strict concavity of $u$, assume by way of contradiction that it is not strictly concave. Then there exists a non-degenerate interval $[a, b]\subset \mathbb{R}$ such that $u$ is affine on it \citep[see][Problems and Remark, A (4), p. 7]{Roberts1973}. Let $\varpi=(a+b)/2$. Take an act $f$ with
\begin{equation}\label{equation: construction of f}
f_1=\varpi+\kappa \mu_2, f_2=\varpi-\kappa \mu_1,  f_i=\varpi \text{ for }i\geq 3,
\end{equation}
so that $f\in [a, b]^N$ for  a sufficiently small and positive number $\kappa$. Since $u$ is affine on $[a, b]$, we have $J(f)=u(\sum_{i=1}^{N}\mu_if_i)=u(\varpi)$ and hence $f\thicksim^*\varpi\mathbf{1}$. By Axiom~\ref{axiom: Translation Invariance 0}, there exists a $\varsigma>0$ such that $g:=\varsigma(f-\varpi\mathbf{1})\succeq^* 0$. 

On the other hand, Axiom~\ref{axiom: Convexity}(b) implies $u$ is not affine on $\mathbb{R}$. Then there exists $\theta^*\in \mathbb{R}$ that is not contained in any nondegenerate interval $J$ on which $u$ is affine. To see this, assume for the sake of contradiction that for every $t \in \mathbb{R}$, there exists a nondegenerate interval $\mathbf{I}_t \subseteq \mathbb{R}$ containing $t$ such that $u$ is affine on $\mathbf{I}_t$. Since $u$ is affine on $\mathbf{I}_t$, there exist constants $a_t, b_t \in \mathbb{R}$ such that $u(r) = a_t r + b_t$ for all $r \in \mathbf{I}_t$. Because $u$ is differentiable on $\mathbb{R}$,  $u'(r) = a_t$ for all $r\in \mathbf{I}_t$. Since $t \in \mathbf{I}_t$, it follows that $u'(t) = a_t$, and thus $u'(r) = u'(t)$ for all $r \in \mathbf{I}_t$. Because $\mathbf{I}_t$ is a nondegenerate interval, there exists at least one rational number $q_t \in \mathbf{I}_t$, so that, $u'(t)=u'(q_t)$. Therefore,  $u'(\mathbb{R})$ is a subset of $u'(\mathbb{Q})$, where $\mathbb{Q}$ is the set of rational numbers. Since $\mathbb{Q}$ is countable, the set $u'(\mathbb{Q})$ is at most countable and hence so is $u'(\mathbb{R})$. However, by Darboux's Theorem \citep[see][Theorem 5-13, pp. 94--95]{Apostol1957},  the derivative $u'$ satisfies the Intermediate Value Property, and therefore $u'(\mathbb{R})$ must be an interval. But the only intervals in $\mathbb{R}$ that contain at most countably many points are degenerate intervals (i.e., singletons). Therefore, $u'(\mathbb{R}) = \{c\}$ for some constant $c \in \mathbb{R}$, which implies  $u$ is globally affine on $\mathbb{R}$, a contradiction.

According to Axiom~\ref{axiom: Translation Invariance 0}, there exist $\alpha>0$ such that $\alpha g+\theta^*\mathbf{1}\succeq^* \theta^*\mathbf{1}$. By the concavity of $u$,  we have
\begin{align*}
  J(\alpha g+\theta^*\mathbf{1}) & = \sum_{i=1}^{N}\mu_i u(\alpha g_i+\theta^*) \leq u(\alpha\sum_{i=1}^{N}\mu_i  g_i+\theta^*)=u(\theta^*).
\end{align*}
This means $\sum_{i=1}^{N}\mu_i u(\alpha g_i+\theta^*)= u(\alpha\sum_{i=1}^{N}\mu_i  g_i+\theta^*)$. Therefore, $u$ must be affine on the convex hull generated by $\{\alpha g_i+\theta^*: i=1, 2, \ldots, N\}$ \citep[see][Theorem~90]{Hardy1934}. From \eqref{equation: construction of f}, we have $g_1>0$ and $g_2<0$, and hence $\theta^*$ is an interior point of that convex hull, a contradiction. This proves the strict concavity of $u$.

We proceed to show  (\rmnum{2}) $\Rightarrow$ (\rmnum{1}). Suppose that  (\rmnum{2}) holds. For notational convenience, let $\mu=(\mu_1, \ldots, \mu_N)$ and
$$V(x)=\sum_{i=1}^N u\left(\langle \delta_i, x\rangle\right)\mu_i.$$
By appealing to the strict increasingness, strict concavity, and continuity of $u$, one can easily verify the validity of Axioms~\ref{axiom: D-monotonicity}, \ref{axiom: Continuity}, \ref{axiom: Convexity}, and \ref{axiom: IUE}.  

It remains to show Axiom~\ref{axiom: Translation Invariance 0}. For any $z^0=(c, 0, 0, \ldots)$, suppose $\alpha x+z^0\succeq z^0$ for some $\alpha>0$ and let $z^{0'}=(c', 0, 0, \ldots)$. Define $U: \mathfrak{F}\rightarrow \mathbb{R}$ by $V=U\circ \phi$; let $g=\phi(x)=(g_1, \ldots, g_N)$ and $\bar g=\sum_{i=1}^N\mu_ig_i.$ By the concavity of $u$ and Jensen's inequality
\begin{equation}\label{equation: jensen}
U(c+\alpha g)=\sum_{i=1}^N\mu_i u(c+\alpha g_i)\leq u\left(c+\alpha\sum_{i=1}^N\mu_i g_i\right)=u(c+\alpha \bar g).
\end{equation}
This along with $\alpha x+z^0\succeq z^0$ yields $ u(c+\alpha\bar g)\geq u(c)$. Since $u$ is strictly increasing, it follows that $\bar{g}\geq 0$. We divide into two cases.

\textbf{Case 1: $\bar g>0$.} Define $H: \mathbb{R}\rightarrow \mathbb{R}$ by
$$H(t)=U(tg+c'\mathbf{1})-U(c' \mathbf{1}).$$
As $u$ is differentiable, so is $H$. A direct calculation yields $H'(0)=u'(c')\bar{g}$, and hence $H'(0)>0$. Therefore, there exists $\tau>0$ such that $H(t)>0$ for all $t\in (0, \tau)$. Choosing any $\alpha'\in (0, \tau)$, we obtain $U(\alpha' g+c'\mathbf{1})>U(c' \mathbf{1})$ or, equivalently, $\alpha' x+z^{0'} \succeq z^{0'}$.

\textbf{Case 2: $\bar g=0$.} Since $\alpha x+z^0\succeq z^0$ and $\bar g=0$, it follows that $U(c+\alpha g)\geq u(c)=u(c+\alpha \bar g)$. By \eqref{equation: jensen}, we get $\sum_{i=1}^N\mu_i u(c+\alpha g_i)= u\left(c+\alpha\sum_{i=1}^N\mu_i g_i\right)$. By the strict concavity of $u$, this implies that there exists a number $\kappa$ such that $g_i=\kappa$ for all $i=1, \ldots, N$, which, together with $\bar g=0$, gives $\kappa=0$. Consequently,  for any $\alpha'>0$, $U(\alpha' g+c'\mathbf{1})=U(c' \mathbf{1})$ or, equivalently, $\alpha' x+z^{0'} \thicksim z^{0'}$. This completes the proof of Axiom~\ref{axiom: Translation Invariance 0}.

\end{appendices}

\bibliographystyle{apa}
\bibliography{library}
\end{document}